\documentclass[a4paper,UKenglish,cleveref, autoref, thm-restate]{lipics-v2021}

\usepackage{algorithm}
\usepackage{algpseudocode}
\algnewcommand\algorithmicforeach{\textbf{for each}}
\algdef{S}[FOR]{ForEach}[1]{\algorithmicforeach\ #1\ \algorithmicdo}
\algrenewcommand\algorithmicprocedure{\textbf{proc}}
\usepackage{booktabs}
\usepackage[T1]{fontenc}
\usepackage{amsthm}

\DeclareMathAlphabet{\mathqhv}{OT1}{qhv}{m}{n}

\title{Efficient construction of the extended BWT from grammar-compressed DNA sequencing reads} 

\titlerunning{Building the eBWT for DNA sequencing read} 

\author{Diego D\'iaz-Dom\'inguez}{CeBiB -- Centre For Biotechnology and Bioengineering, Department of Computer Science, University of Chile, Santiago, Chile \and }{diediaz@dcc.uchile.cl}{https://orcid.org/0000-0002-1825-0097}{ANID Basal Funds FB0001 and Ph.D. Scholarship 21171332, Chile}

\author{Gonzalo Navarro}{CeBiB -- Centre For Biotechnology and Bioengineering, Department of Computer Science, University of Chile, Santiago, Chile}{gnavarro@dcc.uchile.cl}{[orcid]}{ANID Basal Funds FB0001 and Fondecyt Grant 1-200038, Chile}

\authorrunning{D. D\'iaz-Dom\'inguez and G. Navarro} 

\Copyright{Diego D\'iaz and Gonzalo Navarro} 



\category{} 

\relatedversion{} 

\nolinenumbers 

\EventEditors{John Q. Open and Joan R. Access}
\EventNoEds{2}
\EventLongTitle{42nd Conference on Very Important Topics (CVIT 2016)}
\EventShortTitle{CVIT 2016}
\EventAcronym{CVIT}
\EventYear{2016}
\EventDate{December 24--27, 2016}
\EventLocation{Little Whinging, United Kingdom}
\EventLogo{}
\SeriesVolume{42}
\ArticleNo{23}

\begin{document}

\maketitle

\begin{abstract}
We present an algorithm for building the extended BWT (eBWT) of a string collection from its grammar-compressed representation. Our technique exploits the string repetitions captured by the grammar to boost the computation of the eBWT. Thus, the more repetitive the collection is, the lower are the resources we use per input symbol. We rely on a new grammar recently proposed at DCC'21 whose nonterminals serve as building blocks for inducing the eBWT. A relevant application for this idea is the construction of self-indexes for analyzing sequencing {\em reads} --- massive and repetitive string collections of raw genomic data. Self-indexes have become increasingly popular in Bioinformatics as they can encode more information in less space. Our efficient eBWT construction opens the door to perform accurate bioinformatic analyses on more massive sequence datasets, which are not tractable with current eBWT construction techniques. 
\end{abstract}

\section{Introduction}

DNA sequencing reads, or just reads, are massive collections of short strings that encode overlapping segments of a genome. In recent years, they have become more accessible to the general public, and nowadays, they are the most common input for genomic analyses. This poses the challenge of the high computational cost for extracting information from them. Most bioinformatic pipelines must resort to lossy data structures or heuristics, because reads are too massive to fit in main memory~\cite{s15big}. The FM-index's~\cite{ferragina2005indexing} success in compressing and indexing DNA sequences~\cite{li09fast,Langmead2009} opened the door to a new family of techniques that store the full read data compactly in main memory.  This representation is appealing because it retains more information in less space, enabling more accurate biological results~\cite{kaye21the}. One popular FM-index variant for read sets is based on the \emph{extended Burrows-Wheeler Transform} (eBWT)~\cite{mrrs2007ext,bauer13lw,n2007com}. The bionformatics and stringology communities are aware of the potential of the eBWT and have been proposing genomic analysis algorithms on top of eBWT-based self-indexes for years~\cite{cox2012comparing,dolle2017using,guer2019li,pr19snps}.

A bottleneck for the application of those analyses on massive read collections is the high computational cost and memory requirements to build the eBWT. Some authors have proposed semi-external construction algorithms to solve this problem~\cite{bauer13lw,egidi19ext,bon20com}. Still, these techniques slow down rapidly as the input read collection grows. 

The most expensive aspect of computing the eBWT is the lexicographical sorting of the strings' suffixes. In this regard, a recent work~\cite{b2019pr} proposed to speed up the process by exploiting the fact that read sets with high coverage exhibit a good degree of repetitiveness. The goal is to first parse the input to capture long repetitive phrases in the text so as to carry out the suffix comparisons only on the distinct phrases, and then extrapolate the results to the rest of the text. This idea works well for sets of assembled genomes, where repetitiveness is extremely high, but not for reads, where the repetitive phrases are much shorter and scattered through the strings.
Reads sets are much more frequent than fully assembled genomes in applications.

Recently, we presented a fast and space-efficient grammar compressor aimed at read collections~\cite{diaz2021gram}. Apart from showing that grammar compression yields significant space reductions on those datasets, we sketched the potential of the resulting grammar to compute the eBWT directly from it. Similarly to the idea of Boucher \emph{et al.}~\cite{b2019pr}, this method takes advantage of the short repetitive patterns in the reads to reduce the computational requirements. To the best of our knowledge, there are no other repetition-aware algorithms to build the eBWT. An efficient realization of this idea would reduce the computational requirements for indexing reads, thus allowing more accurate genomic analyses on massive datasets.

\textbf{Our contribution.} In this paper we give a detailed description and a parallel implementation of the algorithm for building the eBWT from the grammar of D\'iaz and Navarro~\cite{diaz2021gram}. Our approach not only exploits the repetitions captured by the grammar but also the runs of equals symbols in the eBWT. This makes the time and space per input symbol we require for the construction decrease as the input becomes more repetitive. Our experiments on real datasets showed that our method is competitive with the state-of-the-art algorithms, and that can be the most efficient when the input is massive and with high DNA coverage.

\section{Related concepts}

\subsection{The extended BWT}
Consider a string $T[1..n-1]$ over alphabet $\Sigma[2..\sigma]$, and the sentinel symbol $\Sigma[1]=\texttt{\$}$, which we append at the end of $T$. The \emph{suffix array} (SA)~\cite{MM93} of $T$ is a permutation of $[n]$ that enumerates the suffixes $T[i..n]$ of $T$ in increasing lexicographic order, $T[SA[i]..n] < T[SA[i+1]..n]$. 
The BWT~\cite{bw94} is a permutation of the symbols of $T$ obtained by extracting the symbol that precedes each suffix in $SA$, that is, $BWT[i] = T[SA[i]-1]$ (assuming $T[0]=T[n]=\texttt{\$}$). 
A run-length compressed representation of the BWT
\cite{ma2010s} adds sublinear-size structures that compute, in logarithmic time, the so-called $\mathsf{LF}$ step and its inverse: if $BWT[j]$ corresponds to $T[i]$ and $BWT[j']$ to $T[i-1]$ (or to $T[n]=\textsf{\$}$ if $i=1$), then $\mathsf{LF}(j)=j'$ and $\mathsf{LF}^{-1}(j')=j$. Note that $\mathsf{LF}$ regards $T$ as a circular string.

Let $\mathcal{T}=\{T_1,T_2,...T_m\}$ be a collection of $m$ strings of average size $k$. We then define the string $T[1..n]=T_1\texttt{\$}T_2\texttt{\$}..T_n\texttt{\$}$. The extended BWT (eBWT) of $\mathcal{T}$ \cite{cox2012large} regards it as a set of independent circular strings: the BWT of $T$ is slightly modified so that, if $eBWT[j]$ corresponds to $T_i[1]$ inside $T$, then $\mathsf{LF}(j)=j'$, so that $eBWT[j']$ corresponds to the sentinel \texttt{\$} at the end of $T_i$, not of $T_{i-1}$.

\subsection{Level-Order Unary Degree Sequence (LOUDS).}

LOUDS \cite{j89} is a succinct representation that encodes an ordinal tree $T'$ with $t$ nodes into a bitmap $B[1..2t+1]$, by traversing its nodes in levelwise order and writing down its arities in unary. The nodes are identified by the position where their description start in $B$. Adding $o(t)$ bits on top of $B$ enables constant-time operations like $\textsf{parent}(u)$ (the parent of node $u$), $\textsf{child}(u,i)$ (the $i$-th child of $u$), $\textsf{psibling}(u)$ (the sibling preceding $u$), $\textsf{nodemap}(u)$ (the level-wise rank of node $u$), $\mathsf{leafrank}(u)$ (the number of leaves in level-order up to leaf $u$), $\mathsf{internalrank}(u)$ (the rank of the internal node $u$ in level-order), and $\mathsf{internalselect}(r)$ (the identifier of the \emph{r-th} internal node in level order).

\subsection{Grammar compression}

Grammar compression consists of encoding $T$ as a small context-free grammar $\mathcal{G}$ that only produces $T$. Formally, a grammar is a tuple $(V,\Sigma, \mathcal{R}, \mathqhv{S})$, where $V$ is the set of nonterminals, $\Sigma$ is the set of terminals, $\mathcal{R}$ is the set of replacement rules and $\mathqhv{S} \in V$ is the start symbol. The right-hand side of $\mathqhv{S} \rightarrow C \in \mathcal{R}$ is referred to as the compressed form of $T$. The size of $\mathcal{G}$ is usually measured in terms of $r=|\mathcal{R}|$; the number of rules, $g$; the sum of the length of the right-hand sides of $\mathcal{R}$, and $c=|C|$; the size of the compressed string. The more repetitive $T$ is, the smaller are these values.

\subsubsection{The LMSg algorithm}\label{ssec:lmsg}

$\mathsf{LMSg}$~\cite{diaz2021gram} is an iterative algorithm aimed to grammar-compress string collections. It is based on the concept of induced suffix sorting of Nong \emph{et al}.~\cite{n2009li}. The following definitions of~\cite{n2009li} are relevant for $\mathsf{LMSg}$:

\begin{definition}
A character $T[i]$ is called L-type if $T[i]>T[i+1]$ or if $T[i]=T[i+1]$ and $T[i+1]$ is also L-type. Equivalently, $T[i]$ is said to be S-type if $T[i]<T[i+1]$ or if $T[i]=T[i+1]$ and $T[i+1]$ is also S-type. By default, symbol $T[n]$, the one with the sentinel, is S-type.  
\end{definition}

\begin{definition}
$T[i]$ is called LMS-type if $T[i]$ is S-type and $T[i-1]$ is L-type.
\end{definition}

\begin{definition}
A LMS substring is (i) a substring $T[i..j]$ with both $T[i]$ and $T[j]$ being LMS characters, and there is no other LMS character in the substring, for $i \neq j$; or (ii) the sentinel itself.
\end{definition}

In every iteration $i$, $\mathsf{LMSg}$ scans the input text $T^i$ ($T^1=T$) from right to left to compute the \emph{LMS}-substrings. For each \emph{LMS}-substring $T^i[j..j']$, the algorithm discards the first symbol. If the remaining phrase $F=T[j+1..j']$ has length two or more and at least one of its symbols is repeated in $T^i$, then it records $F$ in a set $\mathcal{D}^{i}$. If $F$ does not meet the condition, then it discards the phrase and inserts its characters into another set $I^i$. Additionally, when $F$ is the suffix-prefix concatenation of two consecutive strings of $T$, $\mathsf{LMSg}$ splits it in two halves, $F_{l}=F[1..u]$ and $F_{r}[u+1..|F|]$, where $F[u]$ contains the dummy symbol. $F_{l}$ and $F_{l}$ are two independent phrases that can be inserted to either $\mathcal{D}^i$ or $I^i$.

After the scan, $\mathsf{LMSg}$ sorts the phrases in $I^i \cup \mathcal{D}^i$ in lexicographical order. If $F \in \mathcal{D}^i$ is a prefix in another phrase $F' \in \mathcal{D}^i$, then the shortest one gets the greatest lexicographical rank (please see \cite{n2009li}). After sorting, the algorithm creates a new rule $\mathqhv{X} \rightarrow F$ for every $F \in \mathcal{D}^i$, where $\mathqhv{X}$ is the sum of $r'$; the highest nonterminal in $V$ before iteration $i$, plus $b$; the lexicographical rank of $F$ in $\mathcal{D}^i \cup I^i$. Every symbol $\mathqhv{Y} \in I^i$ is a nonterminal that already exists in $V$, with a rule $\mathqhv{Y} \rightarrow E$, with $E \in \mathcal{D}^{i'}$ and $i'<i$. Hence, $\mathsf{LMSg}$ updates its value to $\mathqhv{Y}=r'+x$, where $x$ is the lexicographical rank of $\mathqhv{Y}$ in $\mathcal{D}^i \cup I^i$. It also updates the occurrences of $\mathqhv{Y}$ in the right-hand sides of $\mathcal{R}$ to maintain the correctness in the grammar. The last step in the iteration is to replace the partition phrases in $T^i$ with their nonterminal values. This process yields a new text $T^{i+1}$ for the next iteration. The iterations stop when, after scanning $T^i$, no phrases are inserted to $\mathcal{D}^i$. In such case, the algorithm creates the rule for the start symbol as $\mathqhv{S} \rightarrow T^i$. A graphical example of the procedure is shown in Figure \ref{fig:lmsg}A.

\textbf{Post-processing the grammar}. $\mathsf{LMSg}$ collapses the grammar by decreasing every nonterminal $\mathqhv{X} \in V$ to the smallest available symbol $\mathqhv{X'} \notin V \cup \Sigma$.  After the collapse, the algorithm recursively creates new rules with the repeated suffixes of size two in the right-hands of $\mathcal{R}$. These extra rules are called \emph{SuffixPair}. The final grammar for the example of Figure \ref{fig:lmsg}A is shown in Figure \ref{fig:lmsg}B. For simplicity, the symbols are not collapsed in the example.

The $\mathsf{LMSg}$ algorithm ensures the following properties in the grammar to build the eBWT:

\begin{lemma}\label{lem:gram}
For two different nonterminals $\mathqhv{X},\mathqhv{Y} \in V$ produced in the same iteration of $\mathsf{LMSg}$, if $\mathqhv{X}<\mathqhv{Y}$, then the suffixes of $T$ whose prefixes are compressed as $\mathqhv{X}$ are lexicographically smaller than the suffixes whose prefixes are compressed as $\mathqhv{Y}$.
\end{lemma}

\begin{lemma}\label{lem:rank}
Let $S=F[u..|F|]$ be a suffix in a right-hand side $F$ in $\mathcal{R}$. If $|S|>1$ and appears as prefix in some suffix of $T^{i}$, then we can use $S$ to get the lexicographical rank of that suffix among the other suffixes prefixed with sequences other than that of $S$.
\end{lemma}

\begin{definition}\label{def:strind}
The grammar is string independent if the recursive expansion of every $T^{i}[j]$ spans at most one string $T_j \in \mathcal{T}$.
\end{definition}

For further detail in these properties, please see~\cite{diaz2021gram}.

\begin{figure}[h]
\centering
\includegraphics[width=\textwidth]{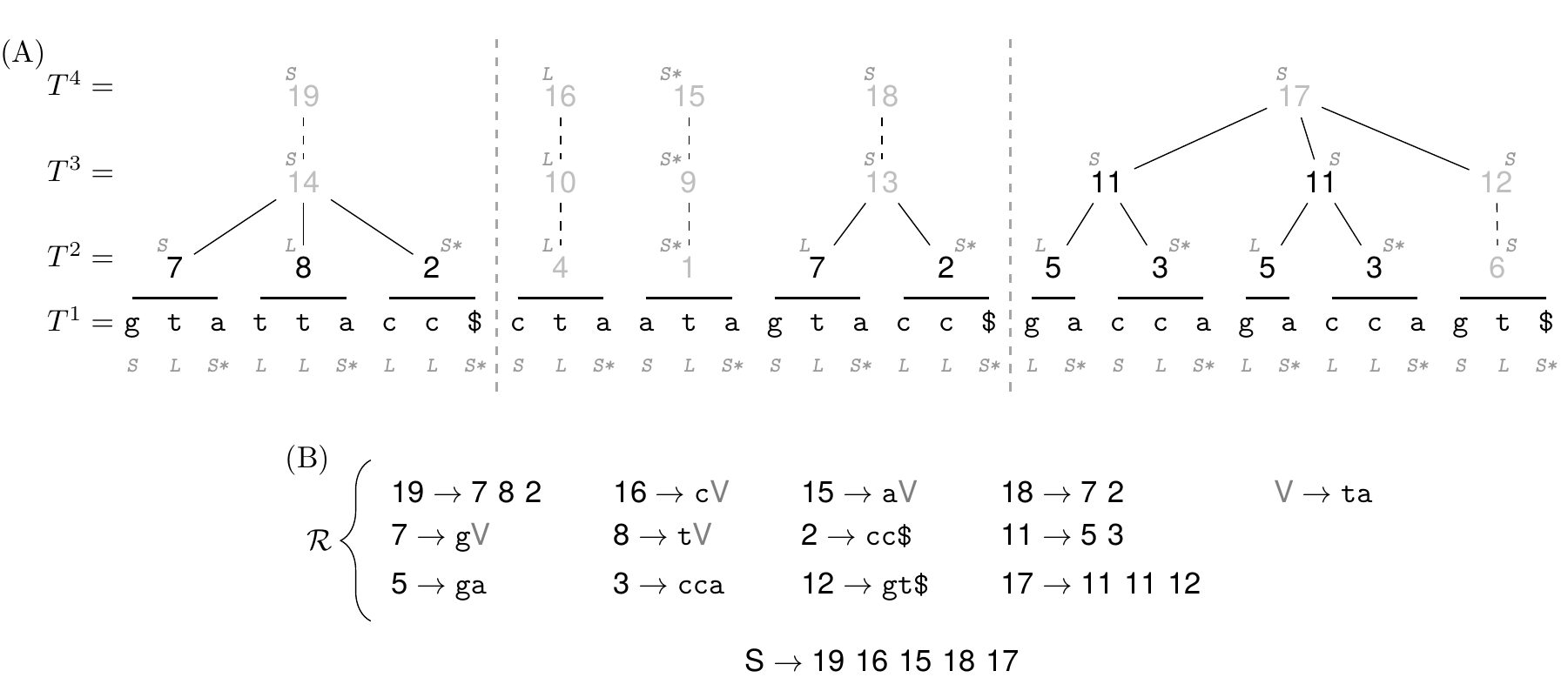}
\caption{(A) Running example of $\mathsf{LMSg}$. The symbols in gray below $T^{1}$ are character types (\emph{L-type}=L, \emph{S-type}=S, \emph{LMS-type}=S*). Dashed vertical lines mark the limits between the strings in $\mathcal{T}$. Every horizontal line on top of $T^{1}$ span one of the phrases generated in the iteration one of $\mathsf{LMSg}$. The parse tree of the grammar is depicted on top of $T^{1}$. Light gray nonterminals have frequency one in $T^{i}$. Dashed edges indicate symbols whose enclosing phrases were discarded for $\mathcal{D}^{i}$. The character at the top of every nonterminal denotes its suffix type. (B) Rules after postprocessing the grammar of (A). For clarity, the nonterminals were not collapsed. Dark gray characters are \emph{SuffPair} nonterminals. The character $\mathqhv{S}$ below $\mathcal{R}$ is the start symbol of $\mathcal{G}$.}

\label{fig:lmsg}
\end{figure}

\textbf{The grammar tree}. D\'iaz and Navarro use a slightly-modified version of the grammar tree data structure of~\cite{cn2012im} to encode $\mathcal{G}$. The construction algorithm is as follows; it starts by creating an empty tree $\mathcal{P}$. Then, it scans the parse tree of $\mathcal{G}$ in level-order. Every time the traversal reaches a new node $y$, the algorithm obtains its label $\mathqhv{X}$. If the extracted symbol is a nonterminal and is the first time it appears, then it creates a new internal node $u$ with $|F|$ children in $\mathcal{P}$, where $F$ is the replacement of $\mathqhv{X}$ in $\mathcal{R}$. The label $l$ of $u$ is the number of internal nodes in $\mathcal{P}$ so far plus $\sigma$, the number of terminals in $\mathcal{G}$. If, on the other hand, $y$ is not the first parse tree node labeled with $\mathqhv{X}$ we visit in the traversal, then the algorithm creates a leaf labeled with $l$ in $\mathcal{P}$ and discards the subtree of $y$. Finally, if $\mathqhv{X}$ is a terminal, then the algorithm creates a new leaf labeled with $\mathqhv{X}$. The shape of $\mathcal{P}$ is then stored using LOUDS, and the leaf labels are stored in a compressed array~\cite{sch964gen}. The internal node labels are not explicitly stored but retrieved on the fly by using the LOUDS operation $\mathsf{internalrank}$. The grammar tree for the grammar of Figure~\ref{fig:lmsg}B is shown in Figure \ref{fig:gt}. The figure also shows how to simulate in $\mathcal{P}$ a traversal over the parse tree of $\mathcal{G}$ by using the LOUDS operations.

\subsubsection{The GLex algorithm}

The grammar tree discards the original nonterminal values, but they are necessary to build the eBWT. D\'iaz and Navarro gave an algorithm called $\mathsf{GLex}$ that reconstructs those values directly from the shape of $\mathcal{P}$. The procedure yields a set of $h$ triplets, where $h$ is the number of iterations of $\mathsf{LMSg}$. Each iteration $i$ has its triplet $(L^i, R^i, f^i)$. The set $L^i$ contains the labels of $\mathcal{P}$ for the phrases in $\mathcal{D}^i \cup I^i$, the set $R^i$ has the lexicographical ranks of those phrases and $f^i$ is a function $f^i(l)=b$ that maps a label $l \in L^{i}$ to its rank $b \in R^i$. For a detailed explanation of the method, we refer the reader to~\cite{diaz2021gram}. 

For simplicity, we consider $R^{i-1}$ to be the alphabet of $T^{i}$ during the construction of the eBWT, altough this is not strictly true. Every nonterminal $\mathqhv{X}=r'+b \in V$ in $T^{i}$ is the sum of $r'$; the highest nonterminal before iteration $i$, and $b$; the lexicographical rank of the phrase $F \in \mathcal{D}^{i-1} \cup I^{i-1}$ to which $\mathqhv{X}$ is assigned. In other words, $\mathsf{GLex}$ only recovers the $b$ part of $\mathqhv{X}$. Still, replacing the nonterimnals of $T^{i}$ with their ranks in $R^{i-1}$ makes no difference for inferring the eBWT as we will see in later sections.

\begin{figure}[t]
\centering
\includegraphics[width=\textwidth]{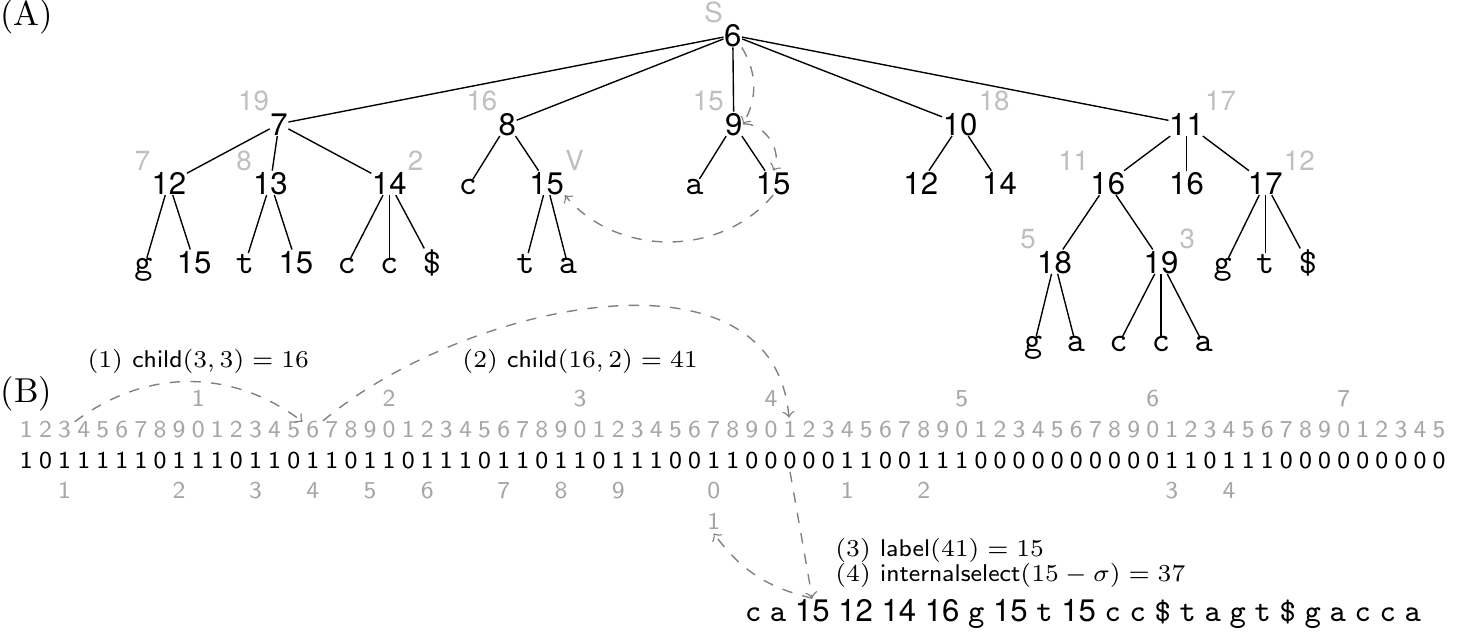}
\caption{ (A) The grammar tree of Figure~\ref{fig:lmsg}B. Numbers on top of the internal nodes are the original nonterminals of the grammar. Dashed arrows simulate a traversal over the parse tree of $\mathcal{G}$ to decompress the word $\texttt{ta}$ from $T^{1}[14..15]$ (Figure~\ref{fig:lmsg}). (B) LOUDS encoding for (A). The bitstream stores the shape of the tree. Gray numbers on top are the bit indexes. Gray numbers below the stream are the internal ranks of the nodes. The integer vector below the stream contains the leaf labels. Dashed arrows mark the same decompression path as in (A), but using the LOUDS functions.}
\label{fig:gt}
\end{figure}

\section{Methods}

Our algorithm for computing the eBWT of $T$ is called $\mathsf{infBWT}$. It first produces the eBWT $B^{h}$ of $C=T^{h}$. Then, it iterates from $h$ to $2$ to infer the eBWT $B^{i-1}$ of $T^{i-1}$ from the already computed transform $B^{i}$. When the iterations finish, $\mathsf{infBWT}$ returns $B^{1}$ as the eBWT of $T$. The overview of the procedure is depicted in Algorithm \ref{algo:infBWT}. Line~\ref{algo:infBWT:inv} is explained later.

\begin{algorithm}[htp]
\caption{Overview of $\mathsf{infBWT}$}
\label{algo:infBWT}{}
\begin{algorithmic}[1]
\Procedure{$\mathsf{infBWT}$}{$\mathcal{P}$} \Comment{returns the eBWT of $T$}
    \State $\mathsf{GLex}(\mathcal{P})$ \Comment{produces the $h$ triplets $(L^{i}, R^{i}, f^{i})$}
    \State Load triplet $h$ from disk
    \State Compute the eBWT $B^{h}$ of $C$ from triplet $h$ and $\mathcal{P}$
    \For{$i=h$ to $2$}
        \State Replace $f^{i}$ with $f^{i}_{inv}$ in triplet $i$\label{algo:infBWT:inv}
        \State Load triplet $i-1$ from disk
        \State $B^{i-1} \gets \mathsf{nextBWT}(B^{i}, f^{i}_{inv}, L^{i-1}, R^{i-1}, f^{i-1})$\label{algo:infBWT:inf}
        \State Discard $B^{i}$ and triplet $i$
        \State $i \gets i-1$
    \EndFor
    \State\textbf{return} $B^{1}$
\EndProcedure
\end{algorithmic}
\end{algorithm}

\subsection{Computing the eBWT of the compressed text}

Unlike the regular BWT, the position of each $C[j]$ in the eBWT does not depend on the whole suffix $C[j+1..c]$, but on the string $S=C[j+1..j+p']C[j-p..j]$. This sequence is a circular permutation of the compressed string $T_{x} \in \mathcal{T}$ encoded in the range $C[j-p..j+p']$. Computing $S$ from the grammar tree $\mathcal{P}$ is simple as $\mathcal{G}$ is string independent (see Definition~\ref{def:strind}). This feature implies that every $T_x \in \mathcal{T}$ maps a specific range $C[k..k']$, with $k\leq k'$. Therefore, we do not have to deal with border cases. For instance, when the compressed suffix or prefix of $T_x$ lies in between two consecutive symbols of $C$.

For constructing the eBWT of $C$ we require $\mathcal{P}$ and $(L^{h}, R^{h}, f^{h})$, the last triplet produced by $\mathsf{GLex}$. Given the definition of $\mathcal{P}$, we can easily obtain the root child $v$ encoding $C[j]$ as $v=\mathsf{nodeselect}(j+1)$. Once we retrieve $v$, we obtain $C[j]$ with $f^h(\mathsf{label}(v))$. For accessing the circular string $S$ to the right of $C[j]$ we define the function $\mathsf{cright}$. This procedure receives as input a position $j \in [1..c]$ and returns another position $j' \in [1..c]$ such that $C[j']$ is the circular right context of $C[j]$. We use $\mathsf{cright}$ as the underlying operator for another function, $\mathsf{ccomp}$. This method compares lexicographically two circular permutations located at different positions of $C$. Similarly, we define a function $\mathsf{cleft}$ that returns the circular left context of $C[j]$. We use it to get the eBWT symbols once we sort the circular permutations. To support these operations, we consider the border cases $C[k'+1]=C[k]$ and $C[k-1]=C[k']$ for every $T_x$. These exceptions require us to include a bitmap $U[1..c]$ that marks as $U[j]=1$ every root child in $\mathcal{P}$ whose recursive expansion is suffixed by a dummy symbol. The functions $\mathsf{cleft},\mathsf{cright}$ and $\mathsf{ccomp}$ are described in Algorithm \ref{algo:cir}.

We start the computation of the eBWT of $C$ by creating a table $A[1..c]$ with $|R^{h}|$ lexicographical buckets. Then, we scan the children of the root of $\mathcal{P}$ from left to right, and for every node $v$,  we store its child rank in the leftmost available cell of bucket $f^{h}(\mathsf{label}(v))$. This process yields a partial sorting of circular permutations of $C$; every bucket $b$ contains the strings that start with symbol $b$. To finish the sorting, we apply a local $\mathsf{quicksort}$ in every bucket, using $\mathsf{ccomp}$ as the comparison function. Notice these calls to $\mathsf{quicksort}$ should be fast as most of the buckets have few elements, and the amount of right contexts we have to access in every comparison is small. Finally, we produce $B^{h}$ by scanning $A$ from left to right and pushing every symbol $f^{h}(\mathsf{label}(\mathsf{nodeselect}(\mathsf{cleft}(A[j])+1)))$ with $j \in [1..|A|]$.

\begin{algorithm}[!ht]
\caption{Functions to simulate circularity over $C$}
\label{algo:cir}{}
\begin{algorithmic}[1]
\Require A bitmap $U[1..|C|]$ marking the symbols of $C$ expanding to phrases suffixed by $\texttt{\$}$.
\Procedure{$\mathsf{cright}$}{$j$} \Comment{returns a $j'$ such that $C[j']$ is the circular right context of $C[j]$}
    \If{$U[j]$}
        \State $j \gets j-1$
        \While{$U[j]$ \textbf{is} $\mathsf{false}$} 
            \State $j \gets j-1$
        \EndWhile
    \EndIf
    \State\textbf{return} $j+1$
\EndProcedure
\Statex
\Procedure{$\mathsf{cleft}$}{$j$} \Comment{returns a $j'$ such that $C[j']$ is the circular left context of $C[j]$}
    \If{$U[j-1]$}
        \While{$U[j]$ \textbf{is} $\mathsf{false}$} 
            \State $j \gets j+1$
        \EndWhile
        \State\textbf{return} $j$
    \Else
        \State\textbf{return} $j-1$
    \EndIf
\EndProcedure
\Statex
\Procedure{$\mathsf{ccomp}$}{$a$,$b$} \Comment{circular lexicographical comparison of $C[a]$ and $C[b]$}
    \State $r_1 \gets f^{h}(\mathsf{label}(\mathsf{nodeselect}(a+1)))$
    \State $r_2 \gets f^{h}(\mathsf{label}(\mathsf{nodeselect}(b+1)))$
    \While{$r_1 \neq r_2$} 
        \State $a \gets \mathsf{cright}(a)$, $b \gets \mathsf{cright}(b)$
        \State $r_1 \gets f^{h}(\mathsf{label}(\mathsf{nodeselect}(a+1)))$
        \State $r_2 \gets f^{h}(\mathsf{label}(\mathsf{nodeselect}(b+1)))$
    \EndWhile
    \State\textbf{return} $r_1<r_2$
\EndProcedure
\end{algorithmic}
\end{algorithm}

\subsection{Inferring the eBWT of the reads}

We define a method called $\mathsf{nextBWT}$ for inferring $B^{i-1}$ from $B^{i}$ (Line~\ref{algo:infBWT:inf} of Algorithm \ref{algo:infBWT}). This procedure requires us to have a mechanism to map a symbol $B^{i}[j]$ to its phrase $F \in \mathcal{D}^{i-1} \cup I^{i-1}$. We support this feature by replacing $f^{i}$ with an inverted function $f^{i}_{inv}$ that receives a rank $b \in R^{i}$ and returns its associated $\mathcal{P}$ label $l \in L^{i}$ (Line~\ref{algo:infBWT:inv} of Algorithm~\ref{algo:infBWT}). Thus, we obtain the grammar tree node $v$ that encodes $F$ with $\textsf{internalselect}(l-\sigma)$. To spell the sequence of $F$ we proceed as follows; we use $\mathcal{P}$ to simulate a pre-order traversal over the subtree of $v$ in the parse tree of $\mathcal{G}$. When we visit a new node $v'$, we check if its label $l'=\mathsf{label}(v')$ belongs to $L^{i-1}$. If that so, then we push $f^{i-1}(l')$ to $F$ and skip the parse subtree under $v'$. The process stop when we reach $v$ again. We call this process the level-decompression of $F$, or $\mathsf{ldc}(v)=F$. 

If we level-decompress all the phrases encoded in $B^{i}$, then we obtain all the symbols of $T^{i-1}$, although not sorted according the eBWT's definition. The position of each decompressed symbol $F[u] \in R^{i-1}$ in $B^{i-1}$ depends, in most of the cases, only on the suffix $F[u+1..|F|]$, except when this suffix has length less than two (Lemma~\ref{lem:rank}). In addition, if two symbols $F[u]$ and $F'[u']$, level-decompressed from different positions $B^{i}[j]$ and $B^{i}[j']$ (respectively), are followed by the same suffix in their respective phrases $F$ and $F'$, then their relative orders in $B^{i-1}$ only depend on the values of $j$ and $j'$. This property is formally defined as follows:

\begin{lemma}\label{lem:bwt1}
Let $B^{i}[j]$ and $B^{i}[j']$ be two eBWT symbols at different positions $j$ and $j'$, with $j<j'$, and whose level-decompressed phrases are $F$ and $F'$, respectively. Also let $S_j$ and $S_{j'}$ be suffixes of $F$ and $F'$ with the same sequence $S$. The occurrence $S_j$ is lexicographically smaller than $S_{j'}$ as it level-decompressed first in $B^{i}$.
\end{lemma}

\begin{proof}
As $S_j$ and $S_{j'}$ are equal, their relative orders depend on the lexicographical ranks of the phrases to the (circular) right of $F$ and $F'$ in the partition of $T^{i-1}$. As $B^{i}[j]$ appears before (from left to right) than $B^{i}[j']$, its right context is lexicographically smaller than the right context of $B^{i}[j']$. Therefore, $S_{j}$ is also lexicographically smaller than $S_{j'}$.
\end{proof}

If we generalize Lemma~\ref{lem:bwt1} to $x\geq1$ occurrences of $S$, then we can use the following Lemma for building $B^{i-1}$:

\begin{lemma}\label{lem:bwt2}
Let $S$ be a string of length at least two, and with $x$ occurrences in $T^{i-1}$. Let $J={j_1, j_2, ...,j_x}$ be the positions of $B^{i}$ encoding the phrases where those occurrences appear as suffixes. Assume we scan $J$ from left to right and for every suffix occurrence of $S$, we extract its left symbol and push it to a list $O$. The resulting list will match a consecutive range $B^{i-1}[o..o']$ of length $o'-o+1=x$. 
\end{lemma}

Finding the range of $O$ in $B^{i-1}$ requires to sort the distinct suffixes in $\mathcal{D}^{i-1} \cup I^{i-1}$ in lexicographical order. Thus, if $S$ has rank $b$, then $O$ is the \emph{b-th} range of $B^{i-1}$. We refer to these suffixes as the \emph{context strings} of the symbols in $B^{i-1}$. The problem, however, is that the suffixes in $\mathcal{D}^{i-1}$ of length one are ambiguous as we cannot assign them a lexicographical order (Lemma~\ref{lem:rank}). We solve this problem by extending the occurrences of the ambiguous context strings with the phrases in $\mathcal{D}^{i-1} \cup I^{i-1}$ that occur to their (circular) right in $T^{i-1}$. By doing this extension, we induce a partition over the $O$ list of an ambiguous $S$; the symbols in first range $O_{X}=O[z..z']$ precede the occurrences of the extended phrase $SX$, the symbols in the second range $O_{Y}=O[z'+1..z'']$ precede the occurrences of another phrase $SY$, and so on. In this way, every segment of $O$ can be placed in some contiguous range of $B^{i-1}$.

We build an FM-index from $B^{i}$ to compute the string extensions, but without including the SA. The idea is that when we extract the occurrence of an ambiguous $S$ from $B^{i}[j]$, we use $\mathsf{LF}^{-1}(j)=j'$ to get the position in $B^{i}$ with the symbol that occurs to the right of $B^{i}[j]$ in $T^{i}$. In this way, we concatenate $S$ to the phrase $F$ obtained from level-decompressing the symbol $B^{i}[j']$. Equivalently, we use $\mathsf{LF}$ to find the left symbols of the context strings that are not proper suffixes of their enclosing phrases. 

When we sort the distinct context strings, we reference them by using nodes of $\mathcal{P}$. If $S$ has length two an appears as suffix in different right-hand sides of $\mathcal{R}$, then there is a \emph{SuffPair} internal node in $\mathcal{P}$ for it (nodes $y2,y3$ and $y4$ of Figure~\ref{fig:infbwt}). Additionally, if $S$ appears as suffix only in one right-hand side, then there is a node (leaf or internal) that we can use as identifier (node $y1$ of Figure~\ref{fig:infbwt}). Finally, if $S$ is not a proper suffix in the right-hand side of the rule, then we use the internal node in $\mathcal{P}$ for that rule (node $v$ of Figure~\ref{fig:infbwt}).

We begin $\mathsf{nextBWT}$ by initializing two empty semi-external lists, $Q$ and $Q'$. Then, we scan $B^{i}$ from left to right. For every $B^{i}[j]$, we first perform $\mathsf{LF}(j)=j'$ to find the position in $B^{i}$ with the symbol that occur to the left of $B^{i}[j]$ in $T^{i}$. Then, we use the functions $f^{i}_{inv}$ to obtain the internal nodes $u$ and $v$ in $\mathcal{P}$ that encode the phrases of $B^{i}[j']$ and $B^{i}[j]$, respectively. Subsequently, we level-decompress the phrase $W \in \mathcal{D}^{i-1} \cup I^{i-1}$ of $u$ and push the pair $(b, v)$ to $Q$, where $b \in R^{i-1}$ is the last symbol of $W$. The next step is to level-decompress the phrase of $v$. During the process, when we reach a node whose label is in $L^{i-1}$, we get its associated symbol $b \in R^{i-1}$ and its right sibling $y$. The next action depends on the value of $l_{y}=\mathsf{label}(y)$. If $l_y$ belongs to $L^{i-1}$ and $y$ is not the rightmost child of its parent, then we push the pair $(b, y)$ to $Q$. If, on the other hand, $y$ is the rightmost child of its parent, but $l_{y}$ does not belong to $L^{i-1}$, then we update its value to $y=\mathsf{internalselect}(l_{y}-\sigma)$ and then push $(b, y)$ into $Q$. The purpose of the $\mathsf{internalselect}$ operation is to get the internal node with the first occurrence of $l_{y}$ in $\mathcal{P}$. The last option is that $l_{y}$ belongs to $L^{i-1}$ and $y$ is the rightmost child of its parent. In such situation, we extend the context of $y$. We use $\mathsf{LF}^{-1}(j)=j''$ to get the position with the symbol to the right of $B^{i}[j]$. As before, we get the internal node $z$ associated to $B^{i}[j'']$ and finally push the triplet $(b, y, z)$ to $Q'$.

Once we complete the scan of $B^{i}$, we group the pairs in $Q$ according to $y$, without changing their relative orders. In $Q'$ we do the same, but we sort the elements according to $(y,z)$. Subsequently, we form a sorted set $U$ with the distinct $y$ symbols of $Q$ plus the distinct pairs $(y,z)$ of $Q'$. To get the relative order of two elements in $U$, we level-decompress their associated phrases and compare them in lexicographical order. For the pairs $(y,z)$, their phrases are obtained by concatenating the level-decompression of $y$ and $z$. Finally, if a given value $y \in Q$  has rank $b$ among the other elements of $U$, then we place the left symbols of its range in $Q$ at the \emph{b-th} range of $B^{i-1}$. The same idea applies for the pairs $(y,z)$ of $Q'$. Algorithm~\ref{algo:inf} implements $\mathsf{nextBWT}$ and Figure~\ref{fig:infbwt} depicts our method for processing $B^{i}[j]$. 

\begin{algorithm}[!ht]
\caption{Inferring $B^{i-1}$ from $B^{i}$}
\label{algo:inf}{}
\begin{algorithmic}[1]
\Require $\mathcal{P}$
\Procedure{$\mathsf{nextBWT}$}{$j$, $f^{i}_{inv}, L^{i-1}, R^{i-1}, f^{i-1}$} \Comment{Produces $B^{i-1}$ from $B^{i}$}
    \State $Q \gets Q' \gets \emptyset$
    \For{$j=1$ to $|B^{i}|$}\label{algo:inf:spar} 
        \State $j' \gets \mathsf{LF}(j)$ \Comment{left symbol of $B^{i}[j]$ in $T^{i}$}
        \State $v \gets \mathsf{internalselect}(f^{i}_{inv}(B^{i}[j])-\sigma)$
        \State $u \gets \mathsf{internalselect}(f^{i}_{inv}(B^{i}[j'])-\sigma)$
        \State $W \gets \mathsf{ldc}(u)$\label{algo:infbwt:lf}
        \State push pair $(W[|W|]], v)$ to $Q$ 
        \If{$\mathsf{label}(v) \notin L^{i-1}$}\Comment{$v$ decompresses to a phrase in $\mathcal{D}^{i-1}$}
            \State $b \gets f^{i-1}(\mathsf{label}(\mathsf{child}(v,1)))$
            \State $y \gets \mathsf{child}(v,2)$
            \While{$\mathsf{true}$}
                \If{$\mathsf{nsib}(y) \neq 0$} \Comment{$y$ is not the rightmost child of its parent}
                    \State push pair $(b, y)$ to $Q$ 
                    \State $b \gets f^{i-1}(\mathsf{label}(y))$
                    \State $y \gets \mathsf{nsibling}(y)$
                \Else
                    \If{$\mathsf{label(y)} \notin L^{i-1}$}\Comment{\emph{SuffPair} node}
                        \State $y \gets \mathsf{internalselect}(\mathsf{label}(y)-\sigma)$
                        \State push pair $(b, y)$ to $Q$ 
                        \State $b \gets f^{i-1}(\mathsf{label}(\mathsf{child}(y,1)))$
                        \State $y \gets \mathsf{child}(y,2)$
                    \Else
                        \State $j'' \gets \mathsf{LF}^{-1}(j)$\Comment{right symbol of $B^{i}[j]$ in $T^{i}$}
                        \State $z \gets \mathsf{internalselect}(f^{i}_{inv}(B^{i}[j''])-\sigma)$
                        \State push triplet $(b, y, z)$ to $Q'$ 
                        \State \textbf{break} 
                    \EndIf
                \EndIf
            \EndWhile
        \EndIf
    \EndFor\label{algo:inf:epar}
    \State $U \gets$ distinct $y$ values of $Q$ $\cup$ distinct $(y,z)$ pairs of $Q'$ 
    \State Sort the strings encoded in $U$\label{algo:inf:sort}
    \State Reorder the elements of $Q \cup Q'$ according the ranks of $U$
    \State $B^{i-1} \gets $ left symbols of $Q \cup Q'$
    \State \textbf{return} FM-index of $B^{i-1}$ 
\EndProcedure
\end{algorithmic}
\end{algorithm}

\subsection{Implicit occurrences of the context phrases}

The $\mathsf{nextBWT}$ algorithm works provided all the occurrences of a context string $S$ appear as suffixes in the phrases of $\mathcal{D}^{i-1} \cup I^{i-1}$. Still, this is not always the case as sometimes they occur in between nonterminals.

\begin{definition}\label{def:impocc}
Let $F=\mathqhv{XYZ}$ be a string in $\mathcal{D}^{i-1}$ with an occurrence in $T^{i-1}[j..j+2]$. This substring is said to be an implicit occurrence of $F$ if, during the parsing of $\mathsf{LMSg}$, $T^{i-1}[j]=\mathqhv{X}$ becomes a suffix of the phrase at $T^{i}[j-p..j]$, with $p\geq1$, and $T^{i}[j+1..j+2]=\mathqhv{YZ}$ is considered another phrase $A \in \mathcal{D}^{i-1}$. 
\end{definition}

An implicit occurrence appears when $F[1]$ is classified as \emph{LMS-type} during the parsing. We use the following lemma to detect this situation:

\begin{lemma}\label{lem:impocc}
The node $v$ encoding $F$ in $\mathcal{P}$ has two children, the left one has a label in $L^{i-1}$ and the right one is a SuffPair node whose label is in $L^{i}$.
\end{lemma}

\begin{proof}
Assume the new rules for $F$ and $A$ after partitioning $T^{i-1}$ are $\mathqhv{F} \rightarrow \mathqhv{XYZ}$ and $\mathqhv{A} \rightarrow \mathqhv{YZ}$, respectively. As $\mathqhv{YZ}$ is a repeated suffix, $\mathsf{LMSg}$ has to create a \emph{SuffPair} nonterminal for it, but it already exists, it is $\mathqhv{A}$. Thus, $\mathsf{LMSg}$ reuses it and replaces $\mathqhv{YZ}$ with $\mathqhv{A}$ in $\mathqhv{F}$'s rule. 
\end{proof}

Before running $\mathsf{nextBWT}$, we scan $L^{i}$ from left to right to find the internal nodes of $\mathcal{P}$ that meet Lemma~\ref{lem:impocc}. For every node $v'$ that meets the lemma, we create a pair $p=(l_{l}, l_{r})$, with the labels in $\mathcal{P}$ of its left and right children, respectively. Then, we record $p$ in a hash table $\mathcal{H}$, with $v'$ as its value. During the execution of $\mathsf{nextBWT}$, when we level-decompress the symbol to the left of $B^{i}[j]$ (Line~\ref{algo:infbwt:lf} of Algorithm~\ref{algo:inf}), we check if the pair formed by the grammar tree label of $W[|W|]$ and the label $f^{i}_{inv}(B^{i}[j])$ has an associated value $v'$ in $\mathcal{H}$. If that happens, then we insert $(W[|W|-1],v')$ to $Q$.

\subsection{Improving the algorithm}

\textbf{Run-length decompression.} As we descend over the iterations of $\mathsf{infBWT}$, the size of $B^{i}$ increases, but its alphabet decreases. This means that in the last steps of the algorithm we end up decompressing a small number of phrases several times. To reduce the time overhead produced by this monotonous decompression pattern, we exploit the runs of equal symbols in $B^{i}$. More specifically, if a given range in $B^{i}[j..j']$ of length $x=j'-j+1$ has the same symbol $\mathqhv{X}$, then we level-decompress the associated phrase once and multiply the result by $x$ instead of performing the same operations $x$ times. By Lemma~\ref{lem:bwt1}, we already now that the repeated symbols resulted from this run-length decompression will be placed in consecutive ranges of $B^{i-1}$. Therefore, when we process a given run $(x, \mathqhv{X})$ of $B^{i}$, we do not push every pair $(b, y)$ $x$ times to $Q$, but insert $(x,b,y)$ only once. We do the same for the triplets $(b,y,z)$ of $Q'$. Figures~\ref{fig:infbwt}B and~\ref{fig:infbwt}C show an example of how the runs of $B^{i}$ are processed.  

\textbf{Unique context strings.} Another way of exploiting the repetitions in $B^{i}$ is with the unique context phrases. Consider a nonterminal $\mathqhv{X}$ encoded in the grammar tree node $v$. If a given child $y$ of $v$ has a label in $L^{i-1}$, then the context phrase $S$ decompressed from $y$ only appears under the subtree of $v$. This means that if $\mathqhv{X}$ has $x$ occurrences in $B^{i}$, then $S$ has $x$ occurrences in $T^{i-1}$, all with the same left context symbol. Computing the elements in $B^{i-1}$ with unique context strings is simple; for each label $l \in L^{i}$, we get the internal node $v=\mathsf{internalselect}(l-\sigma)$, and its number of children $a=\mathsf{nchildren}(v)$. If $a>2$, then we count the $x$ occurrences of $f^{i}(l) \in R^{i}$ in $B^{i}$ by calling $\mathsf{rank}$ over the FM-index. Then, for every child $y$ of $v$ whose child rank is in the range $[2..a-1]$, we push the pair $(x,b,y)$ to $Q$, where $b \in R^{i-1}$ is the symbol obtained from the left sibling of $y$. Then, during the scan of $B^{i}$, we just skip these $y$ nodes. In $\mathsf{infBWT}$, we carry out this process before inverting $f^{i}$ (Line~\ref{algo:infBWT:inv} of Algorithm~\ref{algo:infBWT}). In Figure~\ref{fig:infbwt}A, node $y_1$ in the subtree of $v$ encodes a unique context string.

\textbf{Inferring the eBWT in parallel.}  Our algorithm for building $Q$ and $Q'$ can run in parallel because of Lemma~\ref{lem:bwt2}. For doing so, we divide $B^{i}$ in $p$ chunks and run the for loop between lines \ref{algo:inf:spar}-\ref{algo:inf:epar} of Algorithm~\ref{algo:inf} independently for every one of them. This process yields $p$ $Q$ and $Q'$ lists, which we concatenate as $Q=Q_1 \cup Q_2,..\cup Q_p$ afterward (we do the same with the $Q'$ lists). Then we continue with the inference of $B^{i-1}$ in serial.

\textbf{Sorting the context strings.} There is a trade-off between time and space in sorting $U$ (Line~\ref{algo:inf:sort} of Algorithm~\ref{algo:inf}). On hand side, decompressing the strings once and maintain them in plain form to compare them can produce a considerable space overhead. On the other, decompressing them every time we access them might be slow in practice. We opted for something in between; we only access the strings when we compare them, but we amortize the decompression time overhead by doing the sorting in parallel. Our approach is simple; we use a serial $\mathsf{countingsort}$ to reorder the strings of $U$ by their first symbols. We logically divide the resulting $U$ in buckets; all the strings beginning with the same symbol are grouped in a bucket. To finish the sorting, we just perform $\mathsf{quicksort}$ on parallel in every bucket, decompressing the strings on demand.

\begin{figure}[!htp]
\centering
\includegraphics[width=0.85\textwidth]{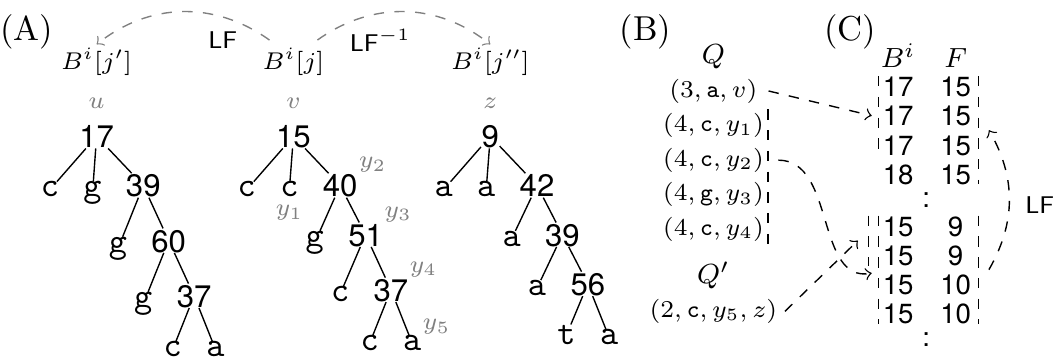}
\caption{(A) Example processing of a symbol $B^{i}[j]=\mathqhv{15}$ by the $\mathsf{nextBWT}$ algorithm. The subtree in the parse tree of $\mathcal{G}$ encoding the phrase $F=\texttt{ccgcca} \in \mathcal{D}^{i-1} \cup I^{i-1}$ of $B^{i}[j]$ is depicted below. The gray symbols are the nodes in $\mathcal{P}$ from which we decode the suffixes of $F$. Node $y_1$ is a unique proper suffix, and nodes $y_{2-5}$ are \emph{SuffPair} suffixes. The left context of $\mathqhv{15}$ in $T^{i}$ is $\mathqhv{17}$, and it is represented in $\mathcal{P}$ with the internal node $u$. Its right context, $\mathqhv{9}$, is represented with the internal node $z$. (B) Pairs inserted to $Q$ and $Q'$ during the processing of $B^{i}[j]$. The leftmost symbols are the frequencies of the suffixes. (C) Processing of $B^{i}[j]$ in run-length mode. Arrays $B^{i}$ and $F$ are the last and first columns of the FM-index, respectively. In $B^{i}$, there is a run for $\mathqhv{15}$ of length 4. Within the run, symbol $\mathqhv{9}$ appears two times as the right context of $\mathqhv{15}$ and $\mathqhv{17}$ appears 3 times as its left context. The $\mathsf{nextBWT}$ algorithm level-decompresses the phrases $\mathqhv{15}$ and $\mathqhv{17}$ only once, and uses the run lengths to include the suffix frequencies to the tuples of $Q$ and $Q'$. This information is represented with dashed lines in the left side of $B^{i}$. Dashed arrows from $Q$ and $Q'$ to $B^{i}$ indicate from which parts of the FM-index the suffix frequencies were extracted.}
\label{fig:infbwt}
\end{figure}

\section{Experiments}

We implemented $\textsf{infBWT}$ as a $\texttt{C++}$ tool called G2BWT. This software uses the $\texttt{SDSL-lite}$ library~\cite{gbmp2014sea} and is available at \url{https://bitbucket.org/DiegoDiazDominguez/lms_grammar/src/bwt_imp2}. We compared the performance of G2BWT against the tools \texttt{eGap} (EG)~\cite{egidi19ext}, \texttt{gsufsort-64} (GS64)~\cite{lou20gsuf} and \texttt{BCR\_LCP\_GSA} (BCR)~\cite{bauer13lw}. EG and BCR are algorithms for constructing the eBWT of a string collection in external or semi-external settings, while GS64 is an in-memory algorithm for building the suffix array of an string collection, but it can also build the eBWT. We also considered the tool \texttt{bwt-lcp-em}~\cite{bon20com} for the experiments. Still, by default it builds both the eBWT and the LCP array, and there is no option to turn off the LCP array, so we discarded it. For BCR, we used the implementation of~\url{https://github.com/giovannarosone/BCR_LCP_GSA}. All the tools were compiled according their authors' description. For G2BWT, we used the compiler flags \texttt{-O3 -msse4.2 -funroll-loops}.

We used five read collections produced from different human genomes~\footnote{\url{https://www.internationalgenome.org/data-portal/data-collection/hgdp}} for building the eBWTs. We concatenated the strings so that dataset 1 contained one collection, dataset 2 contained two collections and so on. All the reads were 152 characters long and had an alphabet of six symbols (\texttt{A,C,G,T,N,\$}). The input datasets are described in Table~\ref{tab:if}.

During the experiments, we limited the RAM usage of EG to at most three times the input size. For BCR, we turned off the construction of the data structures other than the eBWT, and left the memory parameters by default. In the case of GS64, we used the parameter \texttt{--bwt} to indicate that only the eBWT had to be built. The other options were left by default. For G2BWT, we first grammar-compressed the datasets using \texttt{LPG}~\cite{diaz2021gram} an then used the resulting files as input for building the eBWTs. In addition, we let G2BWT to use up to 18 threads. The other tools ran in serial as none of them supported multi-threading.

All the competitor tools manipulate the input in plain form while G2BWT processes the input in compressed space. In this regard, BCR, EG, and GS64 have an extra cost for decompressing the text that G2BWT does not have. To simulate that cost, we compressed the datasets using \texttt{p7-zip} and then measured the time for decompressing them. We assessed the performance of G2BWT first without adding that cost to BCR, EG, and GS64, and then adding it. All the experiments were carried out on a machine with Debian 4.9, 736 GB of RAM, and processor Intel(R) Xeon(R) Silver @ 2.10GHz, with 32 cores.

\begin{table}[htbp]
\centering
\begin{tabular}{cccccc}
\toprule
\textbf{Number of} & \textbf{Plain}    & \textbf{Compressed} &\textbf{\% eBWT} \\
\textbf{collections}    & \textbf{size (GB)}& \textbf{size (GB)}  & \textbf{runs}   \\
\midrule
1             & 12.77     & 3.00            & 31.46                            \\
2             & 23.43     & 5.30            & 26.11                            \\
3             & 34.30     & 7.41            & 22.70                            \\
4             & 45.89     & 9.38            & 20.12                            \\
5             & 57.37     & 11.31           & 18.74                            \\
\bottomrule
\end{tabular}
\caption{Input datasets for building the eBWTs. The compressed size is the space of the \texttt{LPG} representation. The percentage of eBWT runs of a dataset is measured as its number of eBWT runs divided by its total number of symbols, and multiplied by 100.}
\label{tab:if}
\end{table}

\section{Results and discussion}

The results of our experiments without considering the decompression costs for BCR, GS64 and EG are summarized in Figure~\ref{fig:exp1}. The fastest method for building the eBWT was GS64, with a mean time of 0.91 $\mu$secs per input symbol. It is then followed by BCR, G2BWT and EG, with mean times of 0.94, 1.32 and 2.61 $\mu$secs per input symbol, respectively (Figure~\ref{fig:exp1}A). Regarding the working space, the most efficient method was BCR, with an average space of 0.17 bytes of RAM per input symbol. G2BWT is the second most efficient, with an average of 0.78 bytes. EG and GS64 are much more expensive, using 3.07 and 10.54 bytes on average, respectively (Figure~\ref{fig:exp1}B). Adding the decompression overhead to BCR, GS64 and EG increases the running times by 0.02 $\mu$secs per input symbols in all the cases. This negligible penalty is due to $\texttt{p7-zip}$ is fast at decompressing the text. 

Despite G2BWT is not the most efficient method on average, it is the only one that becomes more efficient as the size of the collection increases. While the space and time functions of BCR, EG and GS64 seem to be linear with respect the input size, and with a non-negative slope in most of the cases, in G2BWT these functions resemble a decreasing logarithmic function (see Figures~\ref{fig:exp2}A and \ref{fig:exp2}B). This behavior is due to G2BWT processes several occurrences of a given phrase as a single unit, unlike the other methods. Thus, the cost of building the eBWT depends more on the number of distinct patterns in the input rather than on its size. As genomes are repetitive, appending new read collections to a dataset increases its size considerably, but not its number of distinct patterns. As a consequence, the per-symbol processing efficiency increases.

The performance improvement of G2BWT is also observed when we asses the trade-off between time and space (Figure~\ref{fig:exp1}C). Although BCR is the one with the best trade-off, the instances of G2BWT move toward the bottom-left corner of the graph as we concatenate more read collections. In other words, the more massive and repetitive the input is, the less is the time and space we spend on average per input symbol to build the eBWT. This is an important remark, considering the input collections are not as repetitive as other types of texts. In most of the datasets, the number of eBWT runs is relatively high (see Table~\ref{tab:if}). 

\begin{figure}[htbp]
\centering
\includegraphics[width=\textwidth]{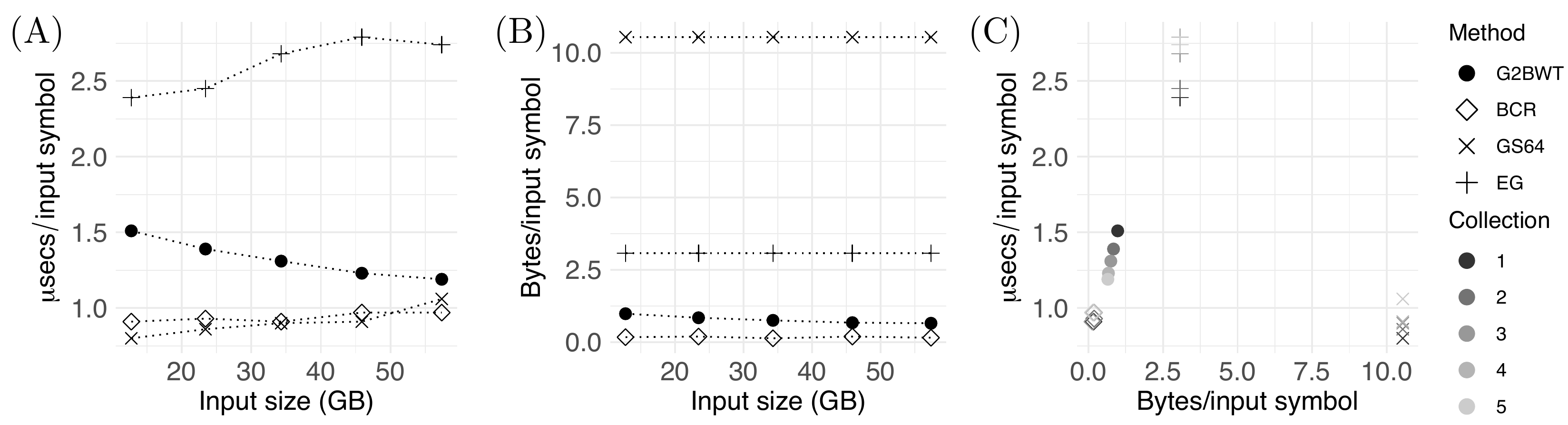}
\caption{Performance of the tools for building the eBWTs. These results do not include the decompression overhead for BCR, GS64, and EG.}
\label{fig:exp1}
\end{figure}

We estimated how many read collections we have to append to achieve a performance equal or better than that of BCR. For that purpose, we built a linear regression for BCR and a logaritmic regression for G2BWT. Subsequently, we checked at which value of $x$ the two models intersect (Figure~\ref{fig:exp2}). For time, we obtained the function $y = 0.887876 + 0.001442x$ for BCR while for G2BWT we obtain $y= 2.0655 - \ln 0.2161x$. The value where these two functions intersect in the x-axis is around $111$ (dashed vertical line of Figure~\ref{fig:exp2}A). Thus, we expect BCR and G2BWT to have a similar performance when the size of the input read dataset is about 111GB. For datasets above that value, we expect G2BWT to be faster.

For the space, the function for BCR was $y=0.1782030 -0.0003511x$ and for G2BWT was $y=1.5574-\ln  0.2277x$. In this case, the functions do not intersect. This is expected as BCR is implemented almost entirely on disk. Therefore, its RAM consumption stays low and stable as the input increases. On the other hand, the RAM usage of G2BWT increases at a slow rate, but it still depends on the input size.

\begin{figure}[h]
\centering
\includegraphics[width=0.7\textwidth]{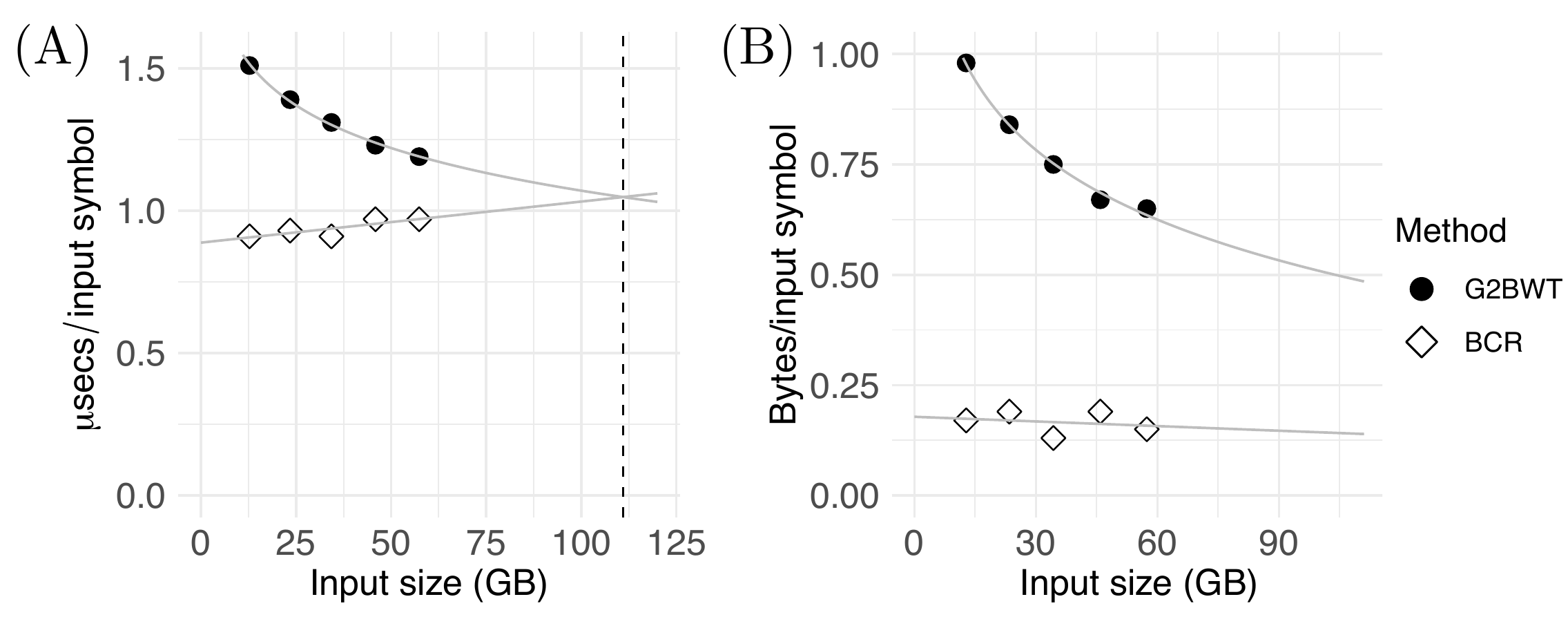}
\caption{Regressions for the performance of BCR and G2BG in Figure~\ref{fig:exp1}. (A) Model for the time. Dashed vertical line indicates where the two functions intersect. (B) Model for the space.}
\label{fig:exp2}
\end{figure}

\section{Concluding remarks}

We introduced a method for building the eBWT that works in compressed space, and whose performance improves as the text becomes massive and repetitive. Our experimental results showed that algorithms that work on top of grammars can be competitive in practice, and even more efficient. Our research is now focused on improving the time of $\mathsf{infBWT}$ and extending its use to build other data structures such as the LCP array. 

\bibliography{main}

\begin{thebibliography}{10}

\bibitem{bauer13lw}
M.~Bauer, A.~Cox, and G.~Rosone.
\newblock Lightweight algorithms for constructing and inverting the {BWT} of
  string collections.
\newblock {\em Theoretical Computer Science}, 483:134--148, 2013.

\bibitem{bon20com}
P.~Bonizzoni, G.~Della~Vedova, Y.~Pirola, M.~Previtali, and R.~Rizzi.
\newblock Computing the multi-string {BWT} and {LCP} array in external memory.
\newblock {\em Theoretical Computer Science}, 2020.

\bibitem{b2019pr}
C.~Boucher, T.~Gagie, A.~Kuhnle, B.~Langmead, G.~Manzini, and T.~Mun.
\newblock Prefix-free parsing for building big {BWT}s.
\newblock {\em Algorithms for Molecular Biology}, 14:article 13, 2019.

\bibitem{bw94}
M.~Burrows and D.~Wheeler.
\newblock A block sorting lossless data compression algorithm.
\newblock Technical Report 124, Digital Equipment Corporation, 1994.

\bibitem{cn2012im}
F.~Claude and G.~Navarro.
\newblock Improved grammar-based compressed indexes.
\newblock In {\em Proc. 19th SPIRE}, pages 180--192, 2012.

\bibitem{cox2012large}
A.~Cox, M.~Bauer, T.~Jakobi, and G.~Rosone.
\newblock Large-scale compression of genomic sequence databases with the
  {B}urrows--{W}heeler transform.
\newblock {\em Bioinformatics}, 28:1415--1419, 2012.

\bibitem{cox2012comparing}
A.~Cox, T.~Jakobi, G.~Rosone, and O.~Schulz-Trieglaff.
\newblock Comparing {DNA} sequence collections by direct comparison of
  compressed text indexes.
\newblock In {\em Proc. 12th WABI}, pages 214--224, 2012.

\bibitem{diaz2021gram}
D.~D\'iaz-Dom\'inguez and G.~Navarro.
\newblock A grammar compressor for collections of reads with applications to
  the construction of the {BWT}.
\newblock In {\em Proc. 31st DCC}, 2021.

\bibitem{dolle2017using}
D.~Dolle, Z.~Liu, M.~Cotten, J.~Simpson, Z.~Iqbal, R.~Durbin, S.~McCarthy, and
  T.~Keane.
\newblock Using reference-free compressed data structures to analyze sequencing
  reads from thousands of human genomes.
\newblock {\em Genome Research}, 27:300--309, 2017.

\bibitem{egidi19ext}
L.~Egidi, F.~Louza, G.~Manzini, and G.~Telles.
\newblock External memory {BWT} and {LCP} computation for sequence collections
  with applications.
\newblock {\em Algorithms for Molecular Biology}, 14:aritcle 6, 2019.

\bibitem{ferragina2005indexing}
P.~Ferragina and G.~Manzini.
\newblock Indexing compressed text.
\newblock {\em Journal of the {ACM}}, 52(4):552--581, 2005.

\bibitem{gbmp2014sea}
S.~Gog, T.~Beller, A.~Moffat, and M.~Petri.
\newblock From theory to practice: Plug and play with succinct data structures.
\newblock In {\em Proc. 13th SEA}, pages 326--337, 2014.

\bibitem{grossi2003}
R.~Grossi, A.~Gupta, and J.~Vitter.
\newblock High-order entropy-compressed text indexes.
\newblock In {\em Proc. 14th SODA}, pages 841--850, 2003.

\bibitem{guer2019li}
V.~Guerrini and G.~Rosone.
\newblock Lightweight metagenomic classification via {eBWT}.
\newblock In {\em Proc. 6th AlCoB}, pages 112--124, 2019.

\bibitem{j89}
G.~Jacobson.
\newblock Space-efficient static trees and graphs.
\newblock In {\em Proc. 30th FOCS}, pages 549--554, 1989.

\bibitem{kaye21the}
Alice Kaye and Wyeth Wasserman.
\newblock The genome atlas: Navigating a new era of reference genomes.
\newblock {\em Trends in Genetics}, 2021.

\bibitem{Langmead2009}
B.~Langmead, C.~Trapnell, M.~Pop, and S.~Salzberg.
\newblock Ultrafast and memory-efficient alignment of short {DNA} sequences to
  the human genome.
\newblock {\em Genome Biology}, 10:article R25, 2009.

\bibitem{li09fast}
H.~Li and R.~Durbin.
\newblock Fast and accurate short read alignment with {B}urrows--{W}heeler
  {T}ransform.
\newblock {\em Bioinformatics}, 25:1754--1760, 2009.

\bibitem{lou20gsuf}
F.~Louza, G.~P Telles, S.~Gog, N.~Prezza, and G.~Rosone.
\newblock gsufsort: constructing suffix arrays, {LCP} arrays and {BWT}s for
  string collections.
\newblock {\em Algorithms for Molecular Biology}, 15:1--5, 2020.

\bibitem{ma2010s}
V.~M{\"a}kinen, G.~Navarro, J.~Sir{\'e}n, and N.~V{\"a}lim{\"a}ki.
\newblock Storage and retrieval of highly repetitive sequence collections.
\newblock {\em Journal of Computational Biology}, 17:281--308, 2010.

\bibitem{MM93}
U.~Manber and G.~Myers.
\newblock Suffix arrays: a new method for on-line string searches.
\newblock {\em SIAM Journal of Computing}, 22:935--948, 1993.

\bibitem{mrrs2007ext}
S.~Mantaci, A.~Restivo, G.~Rosone, and M.~Sciortino.
\newblock An extension of the {B}urrows-{W}heeler {T}ransform.
\newblock {\em Theoretical Computer Science}, 387:298--312, 2007.

\bibitem{n2007com}
G.~Navarro and V.~M{\"a}kinen.
\newblock Compressed full-text indexes.
\newblock {\em ACM Computing Surveys}, 39:article 2, 2007.

\bibitem{n2009li}
G.~Nong, S.~Zhang, and W.~H. Chan.
\newblock Linear suffix array construction by almost pure induced-sorting.
\newblock In {\em Proc. 19th DCC}, pages 193--202, 2009.

\bibitem{pr19snps}
N.~Prezza, N.~Pisanti, M.~Sciortino, and G.~Rosone.
\newblock {SNPs} detection by {eBWT} positional clustering.
\newblock {\em Algorithms for Molecular Biology}, 14:article 3, 2019.

\bibitem{sch964gen}
E.~Schwartz and B.~Kallick.
\newblock Generating a canonical prefix encoding.
\newblock {\em Communications of the {ACM}}, 7:166--169, 1964.

\bibitem{s15big}
Z.~Stephens, S.~Lee, F.~Faghri, R.~Campbell, C.~Zhai, M.~Efron, R.~Iyer,
  M.~Schatz, S.~Sinha, and G.~Robinson.
\newblock Big data: astronomical or genomical?
\newblock {\em PLoS Biology}, 13(7):e1002195, 2015.

\end{thebibliography}
\newpage
\section{Appendix}
\appendix

\section{Representing the FM-index}

Once we infer $B^{i-1}$, we produce a RLFM-index from it. This task is not difficult as the resulting $B^{i-1}$ is actually half-way of being run-length compressed. The only drawback of this representation is that manipulating $B^{i-1}$ can be slow. The RLFM-Index usually represents the BWT using a Wavelet Tree~\cite{grossi2003}. In our case, this feature implies that accessing a symbol in $B^{i}$ and performing $\mathsf{rank}$ has a slowdown factor of $\mathcal{O}(\log r)$. This value can be too slow for our purposes. 
In practice, we use a bit-compressed array to represent $B^{i-1}$ instead of a Wavelet Tree. We also include an integer vector $K$ of the same size of $B^{i-1}$. In every $K[j]$ we store the rank of symbol $B^{i-1}[j]$ up to position $j$. Thus, when we need to perform $\mathsf{LF}$ over $B^{i-1}[j]$, we replace the $\mathsf{rank}$ part in the equation with $K[j]$. Notice it is not necessary to fully load $K$ into main memory as its access pattern is linear. We load it in chunks as we scan $B^{i-1}$ during iteration $i-1$. 

\section{Improving the algorithm even further}

The most expensive aspect of our algorithm is the decompression of phrases. First when we scan $B^{i}$, and then when we sort the suffixes in $U$. Every time we require a string, we extract it on demand directly from the grammar and then we discard its sequence to avoid using extra space. This process is expensive in practice as accessing the grammar several times makes an extensive use of $\mathsf{rank}$ and $\mathsf{select}$ operations. The run-length decompression of the symbols of $B^{i}$ along with the parallel sorting of $U$ try to deal with this overhead, but we can do better.

Consider, for instance, three eBWT runs $\mathqhv{X}_1,\mathqhv{Y}$ and $\mathqhv{X}_2$ placed consecutively at some range of $B^{i}$. In the for loop of $\mathsf{nextBWT}$ (lines \ref{algo:inf:spar}-\ref{algo:inf:epar} of Algorithm~\ref{algo:inf}), both $\mathqhv{X}_1$ and $\mathqhv{X}_2$ are decompressed independently as we are unaware that they are close one each other. We can solve this problem by maintaining a small hash table that record the information of the last visited symbols of $B^{i}$. Thus, when we reach a new eBWT run, we check first if its symbol is in the hash table. If so, then we extract the information of its suffixes from the hash and push it into $Q$ and $Q'$. If not, then we level decompress the run symbol from scratch and store its suffix information in the hash table to avoid extra work in the future. We can limit the size of the hash table so it is always maintained in some of the L1-3 caches, and when we exceed this limit, we just simply reset the hash.

Notice that the copies of the suffixes of $\mathqhv{X}_1$ and $\mathqhv{X}_2$ will be placed consecutively in $B^{i-1}$. We can exploit that fact by not adding the suffix information of $\mathqhv{X}_2$ to $Q$ and $Q'$. Instead, we increase the frequency of the recently added tuples of $\mathqhv{X}_1$ by the length of run of $\mathqhv{X}_2$. In this way, we not only save computing time, but also working space.

Another way of improving the performance is during the sorting of $U$. When we sort the distinct buckets in parallel using $\mathsf{quicksort}$, we can maintain the pivot in plain form, and decompress the strings we compare against it on demand. This change can potentially avoid several unnecessary accessions to the grammar with almost no cost.

\end{document}